\documentclass[]{article}

\usepackage{cite}
\usepackage[linesnumbered,ruled,vlined]{algorithm2e}
\usepackage{amsmath,amssymb,amsbsy,amsgen,amsopn,amsthm,multirow,pgfplots}
\usepackage{textpos}
\usepackage{url}
\usepackage[center]{subfigure}
\usepackage{xspace}
\usepackage[margin=1in]{geometry}

\newtheorem{theorem}{Theorem}
\newtheorem{lemma}{Lemma}

\newcommand{\Zset}{\mathbb{Z}}
\newcommand{\Rset}{\mathbb{R}}

\newcommand{\favg}{w_{\rm avg}}

\newcommand{\supp}{{\rm supp}}

\newcommand{\Hsubmod}{H_{\rm plr}}
\newcommand{\Hcds}{H_{\rm plt}}
\newcommand{\fsubmod}{f_{\rm plr}}
\newcommand{\fcds}{f_{\rm plt}}

\newcommand{\argmin}{\mathrm{argmin}}
\newcommand{\argmax}{\mathrm{argmax}}

\pgfplotsset{compat=1.12}

\graphicspath{{fig/}}

\begin{document}

\title{Adaptive Algorithm for Finding Connected Dominating Sets\\ in
Uncertain Graphs\footnote{This is the accepted version of a paper
to be published by IEEE/ACM
Transactions on Networking.}}

\author{Takuro Fukunaga\thanks{The author is with Chuo University, Japan, 
 JST, PRESTO, Japan, and RIKEN Center for Adcanced Intelligence Project,
 Japan. Email: fukunaga.07s@g.chuo-u.ac.jp. This study was supported by JST PRESTO Grant Number JPMJPR1759.}}

\date{}

\maketitle

 \begin{abstract}
  The problem of finding a minimum-weight connected dominating set (CDS)
  of a given undirected graph has been studied actively, 
  motivated by operations of wireless ad hoc networks.
  In this paper, we formulate a new stochastic variant of the problem.
  In this problem, each node in the graph has a hidden random state,
  which represents whether the node is active or inactive,
  and we seek a CDS of the graph that consists of
  the active nodes.
  We consider an adaptive algorithm
  for this problem,
  which repeat choosing nodes and observing the states of the nodes
  around the chosen nodes until a CDS is found.
  Our algorithms have a theoretical performance guarantee
  that the sum of the weights of the nodes chosen by the algorithm
  is at most $O(\alpha \log (1/\delta))$ times that of any adaptive
  algorithm in expectation, where $\alpha$ is an approximation factor for the
  node-weighted polymatroid Steiner
  tree problem and $\delta$ is the minimum 
  probability of possible scenarios on the node states.
 \end{abstract}

\section{Introduction}
\label{sec.intro}

\paragraph*{Background}
For an undirected graph $G$ with the node set $V$,
a subset $U$ of $V$ is called a \emph{dominating set}
if, for each $v \in V$, $v$ itself or a neighbor of $v$ is included in $U$.
If a dominating set induces a connected subgraph, 
then it is called a \emph{connected dominating set} (CDS).
The CDS problem
is the problem of finding a minimum-cardinality/weight CDS in a given
graph, where the nodes in the graph are weighted in the weight
minimization case.
This problem is being actively studied, 
one motivation being 
operations of wireless ad hoc networks.
For avoiding
flooding of messages in a wireless ad hoc network,
a popular approach is to construct a virtual backbone.
Requirements on the virtual backbone are the connectivity
and 
accessibility from outside the virtual backbone.
Since these requirements are naturally formulated as constraints on CDSs,
it is widely accepted to use a CDS
as the virtual backbone
of a wireless ad hoc network.

In wireless ad hoc networks, nodes are mobile wireless devices, which 
often join or leave the network.
This instability is inevitable because wireless ad hoc networks
are expected to be used in an
environment
for which traditional networks are
hard to use, such as disaster areas.
To cope with the instability, it is desirable that an algorithm
for computing a virtual backbone 
have robustness against node absence.

One approach for constructing a robust virtual backbone 
is to strengthen the requirements on CDSs. In this context, the notion
of a $k$-connected $m$-dominating set ($(k,m)$-CDS) was proposed
by Dai and Wu~\cite{Dai:2006}.
A node set $U$ is called a $(k,m)$-CDS
if it induces a $k$-connected subgraph and each node in $V\setminus U$
has $m$ neighbors in $U$.
There are many previous studies on algorithms for finding a
minimum-cardinality/weight $(k,m)$-CDS (e.g., \cite{Shang:2007jg,ShiZMD17,Wang:2015,Zhang16}).
A disadvantage of this approach is that
the $(k,m)$-CDS tends to be large if $k$
and $m$ are large.
To make a $(k,m)$-CDS tolerant
against clustered node failures (e.g., caused by
an earthquake),
we have to set $k$ and $m$ extremely large.
However, this is infeasible in many cases due to the size disadvantage
of $(k,m)$-CDSs.

\paragraph*{Our contributions}
Another approach 
for the above-stated issue
is to reconstruct CDSs periodically. However, this requires computing
CDSs without knowing which nodes remain in the network.
To cope with this technical challenge, we consider the adaptive
optimization approach, which is recently regarded as a strong paradigm for presenting
robust algorithms against uncertainty (see Section~\ref{subsec.related-adaptive}).
Specifically, we 
formulate a robust optimization version of 
the CDS problem,
that we call the \emph{robust CDS problem}. In this problem, each node in the given
graph has a hidden state, which represents whether the node is active or inactive.
The problem seeks a CDS in the
graph that consists of the active nodes.
For this purpose, we consider adaptive policies for finding a CDS.
An adaptive policy chooses nodes sequentially.
Immediately after a node is chosen,
the policy receives feedback about the states of nodes around the chosen node.
Depending on this feedback, the policy decides which nodes to choose in the
subsequent iterations.
The objective is to minimize the weights of nodes chosen until the
solution set becomes a CDS.

We emphasize that an adaptive policy for the robust CDS problem does not
know which nodes are active when it is invoked.
The only information available to the policy is the probability
distribution on the states of the nodes.
During the process, it
observes them partially through feedback.
We consider two feedback models---the \emph{full feedback model} and the
\emph{local feedback model}.
In the full feedback model,
the states of the neighbors of the chosen node are revealed,
whereas only the state of the chosen node is revealed in the local
feedback model.

We present an algorithm that computes adaptive policies
for the robust CDS problem with both of the feedback models.
We provide a theoretical performance guarantee for this algorithm.
This performance guarantee compares the expected weights of the nodes
chosen by a policy
with those chosen by an optimal adaptive policy.
We prove that their ratio 
is at most $O(\alpha \log (1/\delta))$,
where $\alpha$ is the approximation factor of an algorithm for the
node-weighted polymatroid Steiner tree problem
(which will be defined in Section~\ref{sec.formulation})
and $\delta$ is the minimum 
probability of possible scenarios
on the node states.
If the given graph is a unit disk graph with $n$ nodes,
then we can assume that $\alpha$ is polylogarithmic on $n$ and the
fractionality of the probabilities on the node states
(see Section~\ref{sec.formulation}).
Note that unit disk graphs are popular as a model of wireless 
networks,
and many previous studies on CDSs assume that the graph is a unit disk graph.
In addition to present the performance guarantee,
we investigate the empirical performance of this algorithm through
simulations.

Summing up, the contributions given in this paper are
(i) formulation of the robust CDS problem,
(ii) an algorithm with a theoretical $O(\alpha \log (1/\delta))$-factor 
performance guarantee,
and (iii) comparison of the algorithms through computational simulations.

\paragraph*{Organization}
The rest of this paper is organized as follows.
Section~\ref{sec.related} surveys the related studies.
Section~\ref{sec.formulation} formulates the robust
CDS problem and introduces other preliminary concepts.
Section~\ref{sec.fullfeedback} presents our $O(\alpha \log (1/\delta))$-approximation algorithm.
Section~\ref{sec.heuristic} defines
two heuristic algorithms, with which we compare our algorithm in the simulations.
Section~\ref{sec.simulation} reports the simulation results.
Section~\ref{sec.conclusion} concludes the paper.

\section{Related studies}
\label{sec.related}

\subsection{CDS problem with full information}

Guha and Khuller~\cite{GuhaK98} 
presented 
two $O(H(\Delta))$-approximation algorithms for cardinality minimization
of the CDS problem,
and a $2.613\ln n$-approximation algorithm for weight minimization,
where
$\Delta$ is the maximum degree of the given graph,
$H(\Delta)$ is the $\Delta$-th harmonic number,
and $n$ is the number of nodes of the graph.
They also showed that
no polynomial-time algorithm
for the cardinality minimization 
achieves
approximation factor $(1-\epsilon)H(\Delta)$
for any fixed $\epsilon \in (0, 1)$ unless
${\rm NP} \subseteq {\rm DTIME}[n^{O(\log \log n)}]$.
The approximation factor for weight minimization
was improved by Guha and Khuller~\cite{GuhaK99} to $(1.35+\epsilon)\ln n$.

Besides general graphs, the CDS problem has been considered
extensively for unit disk graphs.
The problem is NP-hard even when the objective is cardinality
minimization and the graph is a unit disk graph~\cite{ClarkCJ90}.
Simultaneously, 
 a polynomial-time approximation scheme (PTAS)
 was given by Cheng~et~al.~\cite{ChengHLWD03}
 for cardinality minimization in a unit disk graph;
note that the existence of a PTAS means
that for
any fixed constant $\epsilon >0$, there exists a 
($1+\epsilon$)-approximation algorithm that runs in polynomial time.
As for weight minimization,
the current best approximation factor
is
$6.475+\epsilon$,
obtained by combining
 Willson~et~al.\ \cite{WillsonZWD15},
Zou~et~al.\ \cite{ZouLGW09},
and Byrka~et~al.~\cite{ByrkaGRS13}.

\subsection{CDS problem with limited information}
One of the algorithms 
presented by Guha and Khuller~\cite{GuhaK98} 
for the cardinality minimization is so-called a
local greedy algorithm, that repeats selecting one or two nodes greedily
from the 2-hop neighborhood of the nodes chosen in the preceding
iterations.
Other local greedy algorithms are also considered 
by Borgs\ et\ al.~\cite{BorgsBCKL12} and 
Khuller and Yang~\cite{KhullerY16}.
Since these local greedy algorithms do not require information on the whole
graph, they work for our setting with a suitable feedback model.
However, the theoretical performance guarantee of these algorithms holds
only for the cardinality minimization, and it seems difficult to extend
them to the weight minimization.

For the cardinality minimization in  unit disk graphs, distributed algorithms are also studied actively.
The current best approximation ratio of distributed algorithms is 6.91 due to Funke\ et\ al.~\cite{funke2006simple}.
Distributed algorithms do not require global information on the graph, and thus their purpose is similar to our study. However,
the distributed algorithms assume that each node in the graph does some computation and communication, which is different from our study. In addition, 
only simple algorithms can be implemented as distributed algorithms.
Indeed, no distributed algorithm is known to the weight minimization.

Shi~et~al.~\cite{ShiCCL16} considers an adaptive algorithm for finding CDSs in energy-harvest sensor networks. The main purpose of this study is 
to manage energy of sensors, and the problem setting considered there is totally different from ours.

\subsection{Adaptive optimization}
\label{subsec.related-adaptive}
We are aware of no previous study on the stochastic version of the
CDS problem, but there are many studies on adaptive 
algorithms for stochastic optimization problems.
One of the previous studies most closely related to our work in this literature
is that by Lim, Hsu, and Lee~\cite{LimHL17} for the 
Bayesian Canadian traveler problem.
The Bayesian Canadian traveler problem
is a stochastic variant of the shortest path problem.
It seeks adaptive policies to find a path from the source to the sink
in the situation where each edge takes the active or inactive state.
The essential idea in \cite{LimHL17} is to consider
the exploration path and the exploitation path in the graph given
by the most-likely scenario repeatedly.
The exploration path is given by an algorithm for the edge-weighted polymatroid
Steiner tree problem,
and the exploitation path is given by an algorithm for the shortest path
 problem.
The algorithm of \cite{LimHL17} compares the weights of the exploration
and the exploitation paths,
and visits nodes along the path that achieves the smaller weight
until the sink is visited or the probability of the scenarios consistent with the observations decreases by a factor of $1/2$.
Similar ideas can be also found in the studies on a
stochastic submodular covering problem with a path constraint \cite{LimHL15},
and the stochastic traveling salesperson problem \cite{GuptaNR17}.

\section{Problem formulation}
\label{sec.formulation}

\subsection{Notations}
$\Rset_+$ and $\Zset_+$ denote the sets of
nonnegative real numbers and nonnegative integers, respectively.
For a positive integer $i$, let $[i]$ denote $\{1,\ldots,i\}$.
Let $G=(V,E)$ denote
an undirected graph with the node set $V$ and the edge set $E$.
Throughout the paper, we denote $|V|$ by $n$.
An edge joining two nodes $u$ and $v$ is denoted
by $uv$.
We say that an edge $uv$ is \emph{induced} by a node set $U \subseteq V$
if $u,v \in U$,
and a subgraph  is induced by a node set $U$
if its node set is $U$ and its edge set consists of the edges induced by $U$.
We denote the subgraph of $G$ induced by $U$ by $G[U]$
and its edge set by $E[U]$.
We say that a subset $U$ of $V$ \emph{dominates} a node $v$
if $v \in U$ or if there exists a node $u \in U$ such that $uv \in E$.
For a node $v \in V$, $N[v]$ denotes the set of neighbors of $v$ and $v$
itself.
For $U\subseteq V$, $N[U]$ is the set of nodes dominated by $U$.

\subsection{CDS problem}
We introduce a formal definition of the CDS problem.
For the connection with the robust CDS problem introduced below,
we present the formulation in which the root node is
specified. 
This formulation is slightly different from the most well-studied setting of the
CDS problem, in which the root node is not specified.
However, the setting without the root can be reduced
to the setting with the root, and hence the formulation given in this
paper is more general.

Suppose that
we are given a connected undirected graph $G=(V,E)$
and the weights $w \colon V \rightarrow \Rset_+$ of the nodes.
Moreover, a node in $V$ is specified as the root node and denoted by $r$.
A subset $U$ of $V$ is called the \emph{connected dominating set} (CDS)
if $U$ includes the root $r$,
induces a connected subgraph of $G$,
and each node in the graph is dominated by $U$.
The weight of $U$ is defined as $w(U):= \sum_{v \in U}w(v)$.
The objective of the problem is to find a minimum-weight CDS.

\subsection{Robust CDS problem}
\label{sec.robustCDS}

In the robust version of the CDS problem, 
several nodes are inactive (absent from the network).
We seek a CDS in the graph consisting of active nodes
for the case that we do not know which nodes are active.
When the graph of the active nodes is not connected, there are two natural
formulations; one formulation demands outputting the message that the
graph is not connected, and the other demands outputting a CDS of the
connected component including the root. We focus 
on the latter formulation in this paper, but our results can be easily applied
to the former formulation with a slight modification.

 \begin{table}
  \caption{List of notations introduced in Section~\ref{sec.robustCDS}}
  \label{table.notations}
  \centering
 \begin{tabular}{|l|l|}
  \hline
  $\Psi$&  set of full realizations\\
  \hline
  $A_{\psi}$ & $\{v \in V\colon \psi(v)=1\}$ for $\psi \in \Psi$\\
  \hline
  $\Psi^*$ & set of realizations\\
  \hline
  $\supp(\phi)$ & $\{v \in V \colon \phi(v)\ne *\}$ for $\phi\in \Psi^*$\\
  \hline
  $U(\pi,\psi)$ & set of nodes chosen by a policy $\pi$ for $\psi \in \Psi$\\
  \hline
  $w(\pi, \psi)$ & $w(U(\pi,\psi))$ for a policy $\pi$ and $\psi \in \Psi$\\ 
  \hline
  $\favg(\pi)$ & expectation of $w(\pi,\psi)$ for a policy $\pi$ and
      $\psi \in \Psi$\\
  \hline
  $\delta$ & $\min\{p(\psi) \colon \psi \in \Psi, p(\psi)>0\}$\\
  \hline
  $M$ & minimum number s.t. $Mp(\psi) \in \Zset_+$ for $\forall \psi
      \in \Psi$\\
  \hline
 \end{tabular}
 \end{table}

Let us formulate the problem more precisely.
The notations introduced in this section is listed in Table~\ref{table.notations}.
We let a vector $\psi \in \{0,1\}^V$
represent which nodes are active or not;
$\psi(v)=1$ (resp., $\psi(v)=0$)
indicates that the node $v$ is active (resp., inactive).
We assume that the root is always active (i.e., $\psi(r)=1$).
We call such a vector $\psi$ a \emph{full realization},
and let $\Psi$ denote the set of all full realizations.
The states of nodes are decided randomly (possibly correlated),
and $p(\psi)\in [0,1]$ is the probability that the states of all nodes are represented by $\psi
\in \Psi$.

For a full realization $\psi$,
let $A_{\psi}=\{v \in V\colon \psi(v)=1\}$.
The task of the robust CDS problem is to find a CDS of the graph
$G[A_{\psi}]$.
If $G[A_{\psi}]$ is not connected, then we seek a CDS of the component
including the root.

For this problem, we consider adaptive policies
for choosing nodes sequentially.
Immediately after the policy chooses a node,
it receives feedback about the states of nodes around the chosen node.
We consider two feedback models.
\begin{itemize}
 \item \emph{full feedback}: if the chosen node $v$ is active,
       then the states of $v$ and neighbors of $v$ are revealed,
       while only the state of $v$ is revealed when $v$ is inactive.
 \item \emph{local feedback}: if a node $v$ is chosen,
       then only the state of $v$ is revealed.
\end{itemize}
If the chosen node is active, then it must be added to the solution set irrevocably, whereas it is discarded otherwise.
We demand that the solution set 
always induce a connected subgraph.
Namely,
if $U$ is the set of nodes chosen by the policy up to a certain iteration,
then the policy has to choose a node 
dominated by $U \cap A_{\psi}$
in this iteration.
The policy repeats this process until it can be determined that
the solution set is a CDS
of the component including the root in $G[A_{\psi}]$.
We assume that the root node has already been included in the
solution set when the first iteration of the policy begins.

We represent the observations made during the algorithm
by a vector $\phi \in \{0,1,*\}^{V}$;
$\phi(v)=1$ (resp., $\phi(v)=0$)
indicates that the state of a node $v$ is observed
and found to be
active (inactive),
and $\phi(v)=*$ indicates that the state of $v$ has not been yet observed by the
policy.
It is assumed that $\phi(r)=1$.
Moreover, in the full feedback model,
$\phi(v)=1$ for each neighbor $v$ of $r$
without loss of generality
because the full feedback model revealed the states of all neighbors of
$r$
and the inactive neighbors of $r$ can be safely removed from the graph.
We call such a vector $\phi$ a \emph{realization},
and let $\Psi^*$ denote the set of all realizations.
We let $\supp(\phi)=\{v \in V\colon \phi(v)\neq *\}$.
We say that a realization $\psi$ \emph{extends} another realization
$\phi$
and write $\psi \succeq \phi$
if $\psi(v)=\phi(v)$ for all $v \in \supp(\phi)$.
For realizations $\psi,\phi \in \Psi^*$ with $\psi \succeq \phi$,
$p(\psi\mid \phi)$ denotes the probability that
the state of each node $v \in \supp(\psi)\setminus \supp(\phi)$
is represented as $\psi(v)$
conditioned on that the states of nodes in $\supp(\phi)$
are represented as $\phi$.

Let $\pi$ be a policy,
and let $U(\pi,\psi)$ be the set of nodes (possibly including inactive nodes)
chosen by $\pi$
when the states of all nodes are represented as $\psi\in \Psi$.
We note that if the observation kept by $\pi$
is $\phi$ when $\pi$ terminates and the full realization is $\psi$,
then $U(\pi,\psi) \subseteq \supp(\phi)$ holds and $U(\pi,\psi)\cap A_{\psi}$ is a
dominating set of $G[\supp(\phi)]$.
Moreover, the CDS output by $\pi$ for a full realization $\psi$
is $U(\pi,\psi) \cap A_{\psi}$.
Particularly
$U(\pi,\psi)= \supp(\phi)$ holds in the local feedback model.
In the full feedback model, all nodes in $U(\pi,\psi)$
are active because a policy knows the state of a node when it is
chosen and it is unnecessary to choose inactive nodes.

We denote $w(U(\pi,\psi))$ by $w(\pi,\psi)$ for conciseness.
Let $\favg(\pi)$ denote the expectation of $w(\pi,\psi)$
(i.e., $\favg(\pi)=\sum_{\psi \in \Psi}w(\pi,\psi)p(\psi)$).
We evaluate the performance of a policy $\pi$ by 
$\favg(\pi)$.
We say that a policy $\pi$ achieves the approximation factor $\alpha \geq 1$
if $\favg(\pi) \leq \alpha \favg(\pi')$ for any policies $\pi'$.

As a parameter, we let
$\delta=\min\{p(\psi) \colon \psi \in \Psi, p(\psi)>0\}$ and 
$M$ denote
the minimum number such that
$M p(\psi)$ is an integer for any $\psi \in \Psi$.
We present an adaptive policy that 
achieves an approximation factor depending on $\delta$ (and the
approximation factor for the node-weighted polymatroid Steiner tree
problem given below).
As $\delta$ becomes smaller, the factor becomes larger.

We would like to emphasize that a policy does not know which nodes
are active in advance.
The information available to the policy when it is invoked is only
the probability distribution $p\colon \Psi \rightarrow [0,1]$ on the full
realizations.
During the process, the policy decides its behavior
from the probability that each node is active and the one that each node is
connected to the root in the graph on the active nodes, conditioned on
the observations obtained up to that point.
If the behavior is independent from the observations,
then the policy is called \emph{nonadaptive}.

\subsection{Polymatroid Steiner tree problem}
\label{sec.polymatroid}
A set-function $f\colon 2^V \rightarrow \Zset_+$ on a finite set $V$
is called \emph{polymatroid}
if it is monotone (i.e., $f(X) \leq f(Y)$ for any $X \subseteq Y
\subseteq V$), submodular (i.e., $f(X \cup \{v\})- f(X) \geq f(Y \cup
\{v\})-f(Y)$ for any $X \subseteq Y \subseteq V$ and $v \in V\setminus
Y$),
and is proper (i.e., $f(\phi)=0$).

In the polymatroid Steiner tree problem, we are given an undirected
graph $G=(V,E)$, a root node $r\in V$,
and a polymatroid $f\colon 2^V \rightarrow \Rset_+$.
A feasible solution is a tree in $G$
such that the root node $r$ is spanned by it,
and the node set $U$ spanned by the tree satisfies $f(U)=f(V)$.
The objective of the problem is to find such a tree of minimum weight.

In the literature, it is usual to assume that
each edge is associated with a nonnegative weight and
the weight of a tree is
defined as the sum of the weights of edges in the tree.
C{\u{a}}linescu and Zelikovsky~\cite{CalinescuZ05}
gave an $O(\frac{1}{\epsilon \log \log n} \log^{2+\epsilon} n \log
f(V))$-approximation algorithm for this case,
where $\epsilon > 0$ is any constant.
In our algorithms for the robust CDS problem, we solve the node-weighted
polymatroid Steiner tree problem, where each node is
associated with a nonnegative weight, and the weight of a tree is defined
as the sum of the weights of the nodes spanned by the tree.
We let $\alpha$ denote the approximation factor of an algorithm for
solving the node-weighted polymatroid Steiner tree problem.
For general graphs, we know no algorithm that achieves a nontrivial
approximation factor for the node-weighted problem.
In the following theorem, we show that, 
if the graph is a unit disk graph,
then the node-weighted problem can be reduced to the edge-weighted
problem,
and hence the approximation factor $\alpha$ is identical to the
one achieved by \cite{CalinescuZ05}.
In unit disk graphs, 
nodes are located on the Euclidean plane,
and
any two nodes are joined by an edge
whenever their distance is at most a unit distance.
From this definition,
it is accepted to model wireless networks as unit disk graphs.

 \begin{theorem}
  \label{thm.polymatroidSteiner}
 If the graph is a unit disk graph, then the node-weighted polymatroid
 Steiner tree admits an $O(\frac{1}{\epsilon \log\log
 n}\log^{2+\epsilon} n \log f(V))$-approximation algorithm
 for any constant $\epsilon > 0$.
 \end{theorem}
  \begin{proof}
   Let $w\colon V\rightarrow \Rset_+$
   be the node weights given in the node-weighted problem.
   Let $\beta$ be the approximation factor for the edge-weighted
   polymatroid Steiner tree problem.
   From $w$, we define edge weights $w' \colon E \rightarrow \Rset_+$
   by $w'(uv)=w(u)+w(v)$ for each $uv \in E$.
   Apply the $\beta$-approximation algorithm 
   for the edge-weighted problem with $w'$.
   Then we obtain a tree $T$ feasible for both the node-weighted and the
   edge-weighted problems. We show that $T$ achieves an
   approximation factor  $O(\beta)$ also for the node-weighted problem.

   Let $T^*$ be an optimal solution for the node-weighted problem.
   It is widely known that,
   for any tree $T^*$ in a unit disk graph, there is a tree $T'$ (possibly the same
   as $T^*$) of maximum degree 5 that spans the same node set as $T^*$
   (see e.g., \cite{Shang:2007jg}).
   Let $w(T')$ and $w'(T')$
   denote the node weight and the edge weight of the tree $T'$.
   We use the same notation also for $T^*$ and $T$.
   Since the maximum degree of $T'$ is at most 5,
   $w'(T')\leq 5 w(T')=5 w(T^*)$.
   Because $T$ achieves an approximation factor $\beta$ for the
   edge-weighted problem,
   we have $w'(T)\leq \beta w'(T')$.
   By the definition of $w'$,  $w(T)\leq w'(T)$. Combining these
   relationships shows that
   $w(T) \leq 5\beta  w(T^*)$.
  \end{proof}

As described above, the polymatroid Steiner tree problem is defined from
a polymatroid. In our algorithms discussed below, we sometimes consider the
problem 
with a function $f$ which is not proper but monotone submodular.
In this case, we can construct a polymatroid $f'$ from $f$ by
$f'(X)=f(X)-f(\phi)$
for each $X \subseteq V$. Indeed, the function $f'$
constructed this way is monotone submodular since the this property is
maintained by the subtraction of a constant, and is proper by $f'(\phi)=f(\phi)-f(\phi)=0$.
Moreover, the problem with $f'$ is equivalent to the
one with $f$
since a node set $U$ satisfies
 $f(U)=f(V)$ if and only if $f'(U)=f'(V)$.

  In instances of the polymatroid Steiner tree problem
  solved in our algorithms,
   $f(V)$ is a polynomial on $n$ and $M$.
   Thus, we can assume that  $\alpha$ is polylogarithmic on $n$ and $M$
   when the graph is a unit disk graph.
   Note that $M\le \prod_{\psi \in \Psi:p(\psi)>0}1/p(\psi)$.
   Thus
   $\log M$ is proportional 
   to the bit-size of the inputs.
   The above upper bound on $M$ is at most $(1/\delta)^{1/\delta}$
   because $|\{\psi \in \Psi:p(\psi)>0\}| \le 1/\delta$. However, we
   have better bounds in special cases. For example, when $p(\psi) \in
   \{0,1/\delta\}$ for all $\psi \in \Psi$, $M=1/\delta$.

\section{Algorithms with performance guarantee}
\label{sec.fullfeedback}

In this section, we present our algorithms
to compute an adaptive policy 
that has a performance guarantee for the robust CDS problem.
We first present the algorithm and its analysis
for the full feedback model. We then present those for the local
feedback model.

Since the algorithm and its analysis
for the local feedback model
are almost the same as those for the full feedback model, we
omit a detailed explanation on the local feedback model and
only highlight
their differences.

\subsection{Algorithm for the full feedback model}
\label{subsec.algorithm-full}
The algorithm runs in rounds. We explain the behavior of the
algorithm
in a round.
Suppose that $\phi\in \Psi^*$ represents the observations made and $U
\subseteq \supp(\phi)$  comprises the nodes chosen in the preceding rounds.
We explain which nodes are chosen by the policy in the current round.

In the algorithm, we safely remove
a node in $V\setminus U$ from the graph if it is
inactive for all full realizations of $\Psi_{\phi}:= \{\psi\in
\Psi\colon \psi \succeq \phi, p(\psi)>0\}$.
Moreover, if a node is disconnected from the
root (because of the removal of inactive nodes),
then this node is not included in and is not dominated by any CDS.
Hence, we also remove such nodes from the graph.
Because of these removals, in the rest of this section,
we suppose that each node in $\supp(\phi)$ is active,
each node in $V \setminus \supp(\phi)$ is active for some full
realization of $\Psi_{\phi}$, and the graph $G$ is always connected.

If $U$ dominates all nodes in the graph (i.e., $\supp(\phi)=V$),
then $U$ is already a CDS.
In such a case, the policy
terminates before entering this round.
In the following, we suppose that $V\setminus \supp(\phi)\neq \emptyset$.
 
\paragraph{Most-likely observations}
We define $H$ as the set of nodes in $V\setminus U$ which is active with probability 
more than a half conditioned on that $\phi$ is realized. In other words, $H=\{v
\in V\setminus U\colon
\sum_{\psi\in \Psi_{\phi}:v\in A_{\psi}} p(\psi\mid \phi) > 1/2\}$. 
Note that all nodes in $\supp(\phi)\setminus U$
belong to $H$ because those nodes are removed from the graph if they are
inactive.
Let $R$ be the set of nodes in the connected component
of the graph $G[U \cup H]$ including $r$;
see also Figure~\ref{fig.graph} for
the definition of $R$.

\begin{figure}
\centering
 \includegraphics[]{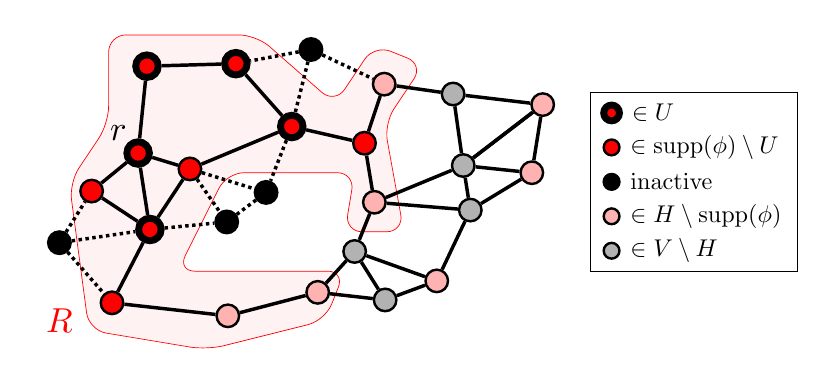}
\caption{An example of the state after a round of the algorithm (this figure also illustrates the inactive nodes removed from the graph)}
 \label{fig.graph}
\end{figure}

For each $v \in R \cap H$,
we define the \emph{most-likely observation vector} $\xi_v \in
\{0,1,*\}^V$ as follows:
\[
 \xi_v(u) =
 \begin{cases}
  * & u \in V \setminus N[v],\\
  1 & u \in H \cap N[v],\\
  0 & u \in N[v]\setminus H\\
 \end{cases}
\]

Our algorithm computes two subsets $\Hcds$ and $\Hsubmod$
of $R \cap H$ from solutions 
to two instances of  the node-weighted polymatroid Steiner tree problem.
In both of these instances, the given graph is $G[R]$,
the root node is $r$,
and the node weights 
$w'\colon R\rightarrow\Rset_+$ are defined by
\[
 w'(v)=
  \begin{cases}
  0 & v \in U,\\
  w(v) & v \in H\setminus U
  \end{cases}
\]
for each $v \in R$.
The polymatroids are different in these instances.
Below, we give
the definitions of these polymatroids
used for computing $\Hcds$ and $\Hsubmod$.

For an intuition,
the polymatroid in the instance for $\Hcds$
is designed so that choosing nodes in $\Hcds$ 
gives a small-weight solution for the robust CDS problem,
assuming that the feedbacks given by choosing nodes in $\Hcds$
are represented as their most-likely observation vectors.
The polymatroid for $\Hsubmod$
is designed so that choosing nodes in $\Hsubmod$ decreases the uncertainty in the node states by at least a half.
Thus, $\Hcds$ aims at the exploitation,
and $\Hsubmod$ aims at the exploration.

\paragraph{Definition of $\Hcds$}
A set-function $\fcds \colon 2^{R} \rightarrow \Rset_+$ used for computing $\Hcds$ is
defined as follows.
Let $H' \subseteq R \cap H$.
The \emph{residual hypothesis}
denoted by $\Psi_{H'}$
is defined as the
set of full realizations $\psi\in \Psi_{\phi}$ such that $\psi\succeq
\xi_v$ for all $v \in H'$.
If two nodes $u$ and $v$ are active and belong to the same connected component
of $G[A_{\psi}]$ in a full realization $\psi \in \Psi$, then we write $u \sim_{\psi} v$.

Now, we define 
$\fcds$ by 
\begin{multline*}
 \fcds(X)
 =|(V\setminus \supp(\phi))\cap N[X]\}|
\\ + \sum_{v \in (V\setminus \supp(\phi)) \setminus N[X]}
\left(1- \sum_{\psi \in \Psi_{X \cap H}: v \sim_{\psi} r}p(\psi)\right)
\end{multline*}
for each $X \subseteq R$.
Thus,
if a node $v \in V\setminus \supp(\phi)$ is dominated by $X$,
then it contributes one unit to $\fcds(X)$.
If $v$ is not dominated by $X$,
then its contribution to $\fcds(X)$
is the probability that $v$ is not active or is not included in the same
component as the root
when most-likely observations
are given as feedback by choosing the nodes in $X\cap H$.
Note that the value of $\fcds(X)$ depends only on $X \cap H$.

We will show below that $\fcds$ is monotone submodular.
\begin{theorem}
\label{thm.submodularity1}
 Function $\fcds$ is monotone submodular.
\end{theorem}

\begin{proof}
 Let $X \subseteq Y \subseteq R \cap H$.
 If a node is dominated by $X$,
 then it is also dominated by $Y$.
 Moreover, if 
 a full realization $\psi \in \Psi$
 satisfies $\psi \in \Psi_{Y}$,
 then also does $\psi \in \Psi_{X}$.
 These imply that $\fcds(X)\le  \fcds(Y)$,
 which shows that $\fcds$ is monotone.

 To see the submodularity of $\fcds$,
 let $u \in (R \cap H)\setminus Y$ and $v \in V \setminus \supp(\phi)$.
 $v$ contributes to $\fcds(X\cup \{u\})-\fcds(X)$
 by $\sum_{\psi \in  \Psi_{X} \setminus \Psi_{X\cup \{u\}}:v \sim_{\psi} r}p(\psi)$
 if $v$ is not dominated by $X \cup \{u\}$,
 and 
 by $\sum_{\psi \in \Psi_{X}:v \sim_{\psi} r}p(\psi)$
 if $v$ is not dominated by $X$ but is dominated by $u$.
 Otherwise, $v$ does not contribute to $\fcds(X\cup \{u\})-\fcds(X)$.
 The value of this contribution is monotone non-increasing on $X$
 because $\Psi_{X}$ is monotone non-increasing on $X$.
 Thus,
 $\fcds(Y\cup \{u\})-\fcds(Y) \leq \fcds(X\cup \{u\})-\fcds(X)$
 follows for any $X \subseteq Y$. Therefore, $\fcds$ is submodular.
\end{proof}

Function $\fcds$ is not a polymatroid
since
it is not proper
(i.e., $\fcds(\emptyset)\neq 0$) and $\fcds$ may not return an integer.
$\fcds$ can be transformed into a proper function by shifting the function values as
mentioned in Section~\ref{sec.formulation}.
As for the integrality of the function values, it suffices to multiply $\fcds$ by $M$.
Let $\fcds'$ be the polymatroid obtained by applying these two operations
to $\fcds$.

In our algorithm, we solve the instance of the polymatroid Steiner
tree problem with $G[R]$, $w'$, and $\fcds'$,
and define $\Hcds$ as the set of nodes in $H$
spanned by the computed tree.
Note that $\fcds(R)  = |V\setminus \supp(\phi)|$ holds.
Hence, $\fcds(\Hcds \cup U)= |V\setminus \supp(\phi)|$ also holds.
Moreover, $\fcds'(R) \leq nM$, and hence the approximation factor
$\alpha$ for solving this instance is polylogarithmic on $n$ and $M$.

\paragraph{Definition of $\Hsubmod$}
We define a set function $\fsubmod \colon 2^{R} \rightarrow
\Rset_+$ by
\[
\fsubmod(X)
=\min\left\{\frac{1}{2},1-\sum_{\psi \in \Psi_{X\cap H}}p(\psi\mid \phi)\right\}
\]
for each $X \subseteq R$.
Again, the value of $\fsubmod(X)$ depends only on $X \cap H$.
The second term in the minimum of this definition
represents the probability that the states of the nodes
are inconsistent with the most-likely vectors
of nodes in $X$, conditioned on that the full realization is consistent
with $\phi$.

\begin{theorem}
\label{thm.submodular2}
 Function $\fsubmod$ is monotone, submodular, and proper.
\end{theorem}
\begin{proof}
 The monotonicity of $\fsubmod$ follows from the fact that 
 $\Psi_{X \cap H}$ is monotone non-increasing on $X$. Moreover, $\fsubmod(\emptyset)=0$
 because
 $\sum_{\psi \in \Psi_{X\cap H}}p(\psi\mid \phi)
 =\sum_{\psi \in \Psi_{\phi}}p(\psi\mid \phi)=1$ holds if $X =\emptyset$.

 To see the submodularity of $\fsubmod$,
 it suffices to see the submodularity of $\fsubmod^{\star} \colon 2^{R\cap H} \rightarrow
 \Rset_+$
 defined by $\fsubmod^{\star}(X)=1-\sum_{\psi \in \Psi_{X}}p(\psi\mid \phi)$
 for each $X \subseteq R \cap H$.
 This is because
 $\fsubmod(X)$ depends only on $X \cap H$ for any $X \subseteq R$,
 and any function $g$ that returns the minimum of a
 constant and the value
 of a function $g^{\star}$
 is submodular if $g^{\star}$ is submodular.
 Let $X \subseteq Y \subseteq R \cap H$ and $v \in (R \cap H)\setminus Y$.
 Then, $\fsubmod^{\star}(X \cup \{v\})- \fsubmod^{\star}(X)=
 \sum_{\psi \in \Psi_{X} \setminus \Psi_{X \cup \{v\}}}p(\psi\mid \phi)$.
 Notice that a full realization is
 included in $\Psi_{X} \setminus \Psi_{X \cup \{v\}}$
 when it is consistent with the most-likely vectors
 of all nodes in $X$ but inconsistent with that of $v$.
 Thus, we have $\Psi_{X} \setminus \Psi_{X \cup \{v\}}
 \supseteq \Psi_{Y} \setminus \Psi_{Y \cup \{v\}}$,
 which implies 
 $\fsubmod^{\star}(X \cup \{v\})-\fsubmod^{\star}(X)
 =\sum_{\psi \in \Psi_{X} \setminus \Psi_{X \cup \{v\}}}p(\psi\mid \phi)
 \geq
\sum_{\psi \in \Psi_{Y} \setminus \Psi_{Y \cup \{v\}}}p(\psi\mid \phi)
 = \fsubmod^{\star}(Y \cup \{v\})-\fsubmod^{\star}(Y)$.
 Therefore, $\fsubmod^{\star}$ is submodular.
\end{proof}

Function $\fsubmod$ does not return an integer value. We can define a
polymatroid $\fsubmod'$ by multiplying $\fsubmod$ by $2Mp(\phi)$.

If $\fsubmod(R) < 1/2$, then we define $\Hsubmod$ as $H$.
Otherwise, we solve the instance of the polymatroid Steiner tree problem
with $G[R]$, $w'$, and $\fsubmod'$,
and define $\Hsubmod$ as the set of nodes in $H$ spanned by
the computed tree.

From the definition of the polymatroid Steiner trees,
$\fsubmod(\Hcds \cup U)=1/2$ holds when $\fsubmod(R)=1/2$.
Note that 
$\fsubmod(\Hcds \cup U)=1/2$ indicates that the probability (conditioned on $\phi$)
of the full realizations
consistent with the most-likely observations
of the nodes in $\Hcds$ is at most a half.
Hence in this case, by choosing the nodes in $\Hcds$ and by observing the
states of their neighbors,
we can decrease the probability of the remaining full realizations
by a factor of at least a half.

\paragraph{Adaptive policy}
After computing $\Hcds$ and $\Hsubmod$, 
we compare their weights and let $H^*$ be the one of smaller weight.
Recall that $\Hsubmod=H$ if $\fsubmod(H \cup U) <1/2$.
Hence,
$H^*=\Hcds$ in this case.

In the adaptive policy, we repeat choosing nodes in $H^*$ in an
arbitrary order such that the chosen nodes and $U$ form a connected
subgraph of $G$. If 
the observation given by choosing a node $v$
is inconsistent with $\xi_v$, then the policy stops choosing nodes and
proceeds to the next round.
Otherwise, the policy continues choosing nodes in $H^*$ unless all nodes have already been chosen.
The pseudo-code of our algorithm for computing this policy is given in Algorithm~\ref{alg.proposal}.

\begin{algorithm}[t]
  \caption{Full feedback model}
 \label{alg.proposal}
 \SetKwInOut{Input}{Input}\SetKwInOut{Output}{Output}
 \Input{undirected graph $G=(V,E)$, node weights $w\colon V\rightarrow
 \Rset_+$, full realizations $\Psi$, probability distribution $p\colon
 \Psi\rightarrow [0,1]$, and root node $r \in V$ ($\psi(r)=1$ for all
 $\psi\in \Psi$)}
  \Output{$U \subseteq V$}
 \BlankLine
 $U \longleftarrow \{r\}$\;
 \ForEach{$v \in V$}{
 \lIf{$v \in N[r]$ and $v$ is active}{$\phi(v) \longleftarrow 1$}
 \lElseIf{$v \in N[r]$}{remove $v$ from $G$}
 \lElse{$\phi(v) \longleftarrow *$}
 }
 \While{$U$ does not dominate $V$}{
 $H \longleftarrow \{v \in V\setminus U\colon \sum_{\psi\in \Psi_{\phi}:v\in
 A_{\psi}} p(\psi\mid \phi) > 1/2\}$\;
 $R \longleftarrow $ connected component including $r$ in $G[H\cup U]$\;
 compute $\Hcds, \Hsubmod \subseteq R \cap H$\;
 $H^* \longleftarrow \argmin_{H' \in \{\Hcds,\Hsubmod\}} w(H')$\;
 \While{$H^* \not\subseteq U$}{
 choose a node $v \in (H^* \setminus U) \cap N[U]$\; 
 $U\longleftarrow U \cup \{v\}$\;
 \lForEach{$u \in N[v]\setminus \supp(\phi)$}{
 $\phi(u)\longleftarrow $ the state of $u$
 }
 \lIf{$\phi(u) \neq \xi_v(u)$ for $\exists u \in N[v]$}
 {break}
 }
 remove $u \in V$ from $G$ if $\psi(u)=0$ for $\forall \psi \in \Psi_{\phi}$\;
 
 remove the connected components not including $r$ from $G$
 }
 output $U$
\end{algorithm}

\paragraph{Example}
\begin{figure}
\centering
 \includegraphics[scale=.8]{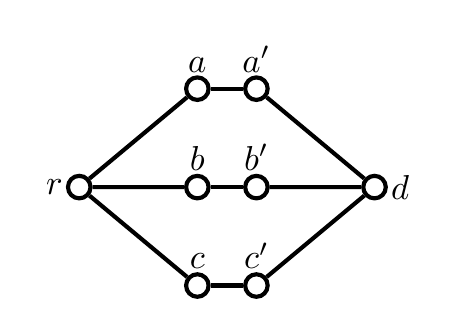}
\caption{An example for illustrating behavior of Algorithm~\ref{alg.proposal}}
\label{fig.example}
\end{figure}

Let us illustrate how Algorithm~\ref{alg.proposal} computes a solution for an example.
Let the given graph $G$ be the one illustrated in Fig.~\ref{fig.example}. 
The root node is $r$,
and exactly one of nodes $a'$, $b'$, $c'$ is active  with probability $1/3$.
(for example, with probability $1/3$, $a$ is active and $b,c$ are inactive).
All of the other nodes are active with probability $1$.
The weights of $a,b,c$ are defined so that $w(a) < w(b)<w(c)$.
The weights of the other nodes are all 0.
We explain behavior of Algorithm~\ref{alg.proposal} for this instance.

Initially, $U=\{r\}$, and the algorithm knows that $r,a,b,c,d$  are active.
Since none of $a',b',c'$ is active with probability at least 1/2, $H=\{a,b,c,d\}$
and $R=\{r,a,b,c\}$.
In the polymatroid Steiner tree problem with $\fcds$,
the constraint demands that the node set spanned by the solution dominates all the nodes whose states are not revealed (i.e., $a',b',c'$). Thus, $\Hcds=R$. The constraint of the problem with $\fsubmod$ demands that 
at least two of $a',b',c'$ are dominated by the node set spanned by the solution.
Hence, $\Hsubmod=\{r,a,b\}$  if we solve the polymatroid Steiner tree problem optimally
(here, we assume that an exact algorithm for the polymatroid Steiner tree problem is used).
Then, the algorithm set $H^*$ to $\Hsubmod$, and the policy chooses $a$ and $b$ sequentially.
Suppose that $a$ is chosen first. If $a'$ is revealed to be active, the policy terminates the first iteration.
Then, it chooses $a'$ in the next iteration and outputs $\{r,a,a'\}$ as a solution in this case.
If $a'$ is inactive, then the policy chooses $b$ and observes the state of $b'$.
If $b'$ is active, then the policy chooses $b'$ in the second iteration and outputs $\{r,a,b,b'\}$ as a solution.
If $b'$ is inactive, then the policy chooses $c,c'$ in the second iteration, and outputs $\{r,a,b,c,c'\}$ as a solution.
This policy $\pi$ achieves $\favg(\pi)= w(a) + 2/3\cdot w(b) + 1/3 \cdot w(c)$.

\subsection{Performance guarantee for the full feedback model}
\label{subsec.guarantee}

Let $\pi$ denote the adaptive policy given by
Algorithm~\ref{alg.proposal}.
In this section, we prove the following performance guarantee on $\pi$.

 \begin{theorem}
  \label{thm:approx_ratio}
  Let $\pi$ denote the policy given by Algorithm~\ref{alg.proposal},
  and let $\alpha$ be the approximation factor for the node-weighted polymatroid Steiner
  tree problem.
  Then, $\pi$ achieves the approximation factor $2\alpha  (1+\lg (1/\delta))$.
 \end{theorem}

Let $U \subseteq V$ and $\phi \in \Psi^*$ such that $\supp(\phi)=U \cup N[U]$.
We consider a residual instance 
of the robust CDS problem after choosing nodes in $U$ and the obtained
observations are represented by $\phi$.
That is, when a policy for this instance is invoked,
the nodes in $U$ have been already chosen and the states of those nodes and their neighbors 
are given as $\phi$. Let $\pi^*(U,\phi)$ be an optimal policy
for this residual instance.
Note that $\favg(\pi^*(U,\phi))$ does not count the weights of nodes in $U$.

 \begin{lemma}
  \label{lem.cost}
  Let $U$ and $\phi$ be those at the beginning of some round of $\pi$,
  and let $H^*$ denote the one computed in this round. Then,
  $w(H^*) \leq 2\alpha \favg(\pi^*(U,\phi))$.
 \end{lemma}
  \begin{proof}
   In this proof, we simply denote $\pi^*(U,\phi)$ by $\pi^*$.
   Let $(v_1,\ldots,v_l)$ be the sequence of nodes chosen by $\pi^*$
   when it always receives feedback of the highest probability,  and for each $i\in [l]$, 
   let $\phi_i \in \Psi^*$ represent the realization
   maintained after choosing $v_i$ and observing the states of the neighbors of $v_i$.
   The assumption on feedback means that
   $\phi_{i}$ maximizes $\sum_{\psi\in \Psi_{\phi_{i}}}
   p(\psi\mid \phi_{i-1})$ over possible realizations, 
   where, by convention, $\phi_0$ denotes $\phi$.
   Let $j$ be the maximum index such that
   $\sum_{\psi\in \Psi_{\phi_j}} p(\psi\mid \phi) >
   1/2$.
   The maximality of $j$ implies
   $\sum_{\psi\in \Psi_{\phi_{j+1}}} p(\psi \mid \phi) \leq
   1/2$  if $j < l$.
   We let $V_i$ denote the node set $\{v_1,\ldots,v_i\}$ for each $i\in [l]$.

   Firstly we show that $v_i \in H$ for each $i \in [j+1]$.
   Recall that any policy in the full feedback model always chooses a
   node which is known to be active. Hence, $\phi_{i-1}(v_i)=1$.
   This indicates that any $\psi \in \Psi_{\phi_{i-1}}$ 
   satisfies $v_i \in A_{\psi}$.
   Moreover, $\Psi_{\phi_{j}}\subseteq \Psi_{\phi_{i-1}} \subseteq \Psi_{\phi}$ because
   $\phi_j \succeq \phi_{i-1}\succeq \phi$.
   Therefore,
   $\sum_{\psi\in \Psi_{\phi}:v_i \in A_{\psi}}p(\psi\mid \phi) \geq
   \sum_{\psi\in \Psi_{\phi_{i-1}}}p(\psi\mid \phi) 
   \geq \sum_{\psi\in \Psi_{\phi_{j}}}p(\psi\mid \phi) 
   > 1/2$,
   which means that $v_i \in H$.

   Secondly, we prove that the observations given by
   choosing $v_i$ are consistent with $\xi_{v_i}$
   if $i \in [j]$. For this, 
   let $x \in \{0,1,*\}^V$ represent the 
   states of the neighbors of $v_i$ observed at the choice of $v_i$
   ($x(u)=*$ if $u\not\in N[v_i]$,  $x(u)=1$ if $u \in N[v_i]$ and $u$ is active,
   and $x(u)=0$ otherwise)
   and suppose that $x(u) \neq \xi_{v_i}(u)$ for some $u \in N[v_i]$.
   From the definitions of $\xi_{v_i}$ and $H$, we have
   $\sum_{\psi\in \Psi_{\phi}: \psi(u)=x(u)}p(\psi\mid \phi) =
   \sum_{\psi\in \Psi_{\phi}: \psi(u)\neq \xi_{v_i}(u)}p(\psi\mid \phi) \leq 1/2$,
   from which 
   $\sum_{\psi\in \Psi_{\phi}: \psi\succeq x}p(\psi\mid \phi)\leq 1/2$ follows.
   However, this means that 
   $\sum_{\psi\in \Psi_{\phi_{i}}}p(\psi\mid \phi)\leq 1/2$ holds
   because $\sum_{\psi\in \Psi_{\phi_{i}}}p(\psi\mid \phi)\leq 
   \sum_{\psi\in \Psi_{\phi}: \psi\succeq x}p(\psi\mid \phi)$.
   This contradicts the definition of $j$, and hence $x = \xi_{v_i}$.
   
  Thirdly, we show that $w(H^*) \leq \alpha w(V_{j})$ holds when $j=l$.
   For this, we suppose $j=l$ for a moment.
   The union of $U$ and $V_l=\{v_1,\ldots,v_l\}$ 
   is a feasible solution for the instance of the polymatroid Steiner tree problem solved for computing $\Hcds$.
   Indeed,
   $V_l \cup U$ induces a connected
   subgraph of $G[R]$.
   Moreover, when $\pi^*$ terminates,
   each node in $V\setminus \supp(\phi)$
   is dominated by $V_l \cup U$, is revealed to be inactive,
   or is disconnected from the root
   in any graphs consistent with the observations.
   This means that $\fcds(V_l \cup U)=|V \setminus \supp(\phi)|=\fcds(R)$, and thus 
   $V_l \cup U$ is feasible for the instance.
   Therefore, $\alpha w(V_l) \geq w(\Hcds) \geq w(H^*)$ in this case.

   Next, we prove that 
   $w(H^*) \leq \alpha w(V_{j+1})$ holds when $j <l$.
   Suppose that $j < l$.
   If the observations given at the choice of $v_{j+1}$
   are consistent with $\xi_{v_{j+1}}$,
   then $\sum_{\psi\in \Psi_{V_{j+1}}}p(\psi\mid \phi)=\sum_{\psi\in \Psi_{\phi_{j+1}}}p(\psi \mid \phi)$.
   Otherwise,
   $\sum_{\psi\in \Psi_{V_{j+1}}}p(\psi\mid \phi)\leq\sum_{\psi\in \Psi_{\phi_{j+1}}}p(\psi \mid \phi)$ from the definition of $\phi_{j+1}$.
   In either of the cases,
   $\sum_{\psi\in \Psi_{V_{j+1}}}p(\psi\mid \phi)\leq\sum_{\psi\in \Psi_{\phi_{j+1}}}p(\psi \mid \phi)$,
   and the right-hand side of this inequality is at most $1/2$ because of the definition of $j$.
   This means that $\fsubmod(V_{j+1} \cup U)=1/2=\fsubmod(R)$, and 
   $V_{j+1} \cup U$ is feasible for the instance of the polymatroid Steiner tree problem
   solved for computing $\Hsubmod$.
   This implies $\alpha w(V_{j+1})\geq w(\Hsubmod)\geq w(H^*)$.
   
   Now, to prove the lemma, it suffices to show $w(V_{l}) \leq 2\favg(\pi^*)$
   if $j = l$,
   and 
   $w(V_{j+1}) \leq 2\favg(\pi^*)$
   otherwise.
   From the definition of $j$,
   $\phi_{i}$ appears with probability at least $1/2$ in the execution of $\pi^*$ for each $i \in [j]$.
   Hence, 
    with probability at least $1/2$, 
   policy $\pi^*$ chooses the nodes in $V_{j}$,
   and $v_{j+1}$ when $j < l$.
   Therefore, the required inequalities hold.
  \end{proof}

  \begin{lemma}
\label{lem.numround}
   The number of rounds executed in $\pi$ is at most $1+\lg (1/\delta)$.
  \end{lemma}
  \begin{proof}
   Let $\phi$ and $\phi'$ be the realizations at the beginning and the end of a round.
   We prove that $p(\phi)/2 \geq p(\phi')$ holds unless it is the last round.
   This indicates that the number of rounds is $1 + \lg (1/\delta)$.

   First, let us consider the case where an observation given at the choice of a node 
   $v \in H^*$ is inconsistent with $\xi_v$ in this round.
   By the definition of $\xi_v$,
   $\sum_{\psi \succeq \Psi_{\phi}\colon \psi \not\succeq \xi_v}p(\psi\mid \phi) \leq 1/2$ holds.
   Since 
   $\phi'$ is a realization in $\Psi_{\phi}$
   such that $\phi'\not\succeq \xi_v$ in this case, 
   we have $p(\phi')\leq p(\phi)/2$.
   
   In the following,
   suppose that the observation given at the choice of each node 
   $v \in H^*$ is consistent with $\xi_v$.
   Thus, all nodes in $H^*$ are chosen in this round.
   If $H^*=\Hcds$,
   then choosing all nodes in $\Hcds$ indicates that the set of nodes
   chosen up to the end of this round dominates the connected component
   of active nodes including the root, and hence this round is the last.
   Otherwise (i.e., $H^*=\Hsubmod$), $p(\phi'\mid \phi) \leq 1/2$ because $\phi'$
   is consistent with $\xi_v$ for all $v \in \Hsubmod$
   and $\fsubmod(\Hsubmod \cup U)=1/2$.
  \end{proof}

  \begin{lemma}
\label{lem.opt_part}
   Let $\{\Phi_1,\ldots,\Phi_k\}$ be a partition of $\Psi$.
   For each $i \in [k]$, 
   let $p_i=\sum_{\psi \in \Phi_i}p(\psi)$,
   consider 
   an instance such that 
   each $\psi \in \Phi_i$ appears with probability $p(\psi)/p_i$,
   and let $\pi^*_i$ be an optimal policy for this instance.
   Moreover, $\pi^*$ denotes an optimal policy for the original instance.
   Then, $\sum_{i\in [k]} p_i \favg(\pi^*_i) \leq \favg(\pi^*)$.
  \end{lemma}
  \begin{proof}
   Recall that $U(\pi^*,\psi)$ is the output of $\pi^*$ for a full realization
   $\psi$.
   Define $\pi_i$ as the policy for the instance with $\Phi_i$
   such that $U(\pi_i,\psi)=U(\pi^*,\psi)$ for each $\psi \in \Phi_i$.
   By the optimality of $\pi^*_i$, 
   $\favg(\pi^*_i)\leq \favg(\pi_i)=\sum_{\psi\in \Phi_i} U(\pi_i,\psi)p(\psi)/p_i$.
   Hence,
   $\sum_{i \in [k]}p_i \favg(\pi^*_i)\leq \sum_{i\in [k]}
   \sum_{\psi\in \Phi_i} U(\pi^*,\psi)p(\psi)=\favg(\pi^*)$.
  \end{proof}

\begin{proof}[Proof of Theorem~\ref{thm:approx_ratio}]
 We show that the approximation factor of $\pi$ is $2\alpha k$ if the number of rounds
 in $\pi$ is $k$. Since $k\leq 1+\lg (1/\delta)$ by Lemma~\ref{lem.numround},
 this suffices to prove the theorem.

 Let $\pi^*$ be an optimal policy for the given instance.
 The weights of nodes chosen in the first round of $\pi$ are at most $2\alpha\favg(\pi^*)$
 by Lemma~\ref{lem.cost}.
 Let $(U_1,\phi_1),\ldots,(U_l,\phi_l)$ be the pairs of possible solutions and the observations
 at the end of the first round.
 Then, $\Psi_{\phi_1},\ldots,\Psi_{\phi_l}$ form a partition of $\Psi$.
 The number of rounds of $\pi$ 
 in the residual instance with $(U_i,\phi_i)$
 is at most $k-1$ for each $i \in [l]$.
 Hence, 
 under the condition that the full realization extends $\phi_i$,
 the expected objective value achieved by $\pi$ in this residual instance
 is at most $2\alpha(k-1)\favg(\pi^*(U_i,\phi_i))$.
 This means that 
 \begin{align}
  \label{eq.induction}
  &\favg(\pi)  \nonumber \\
  &\leq 2\alpha \favg(\pi^*) 
  + 2\alpha(k-1) \sum_{i=1}^l p(\phi_i) \favg(\pi^*(U_i,\phi_i))
 \end{align}
Note that $\favg(\pi^*(U_i,\phi_i))$ is at most the expected objective value
 achieved by an optimal policy $\pi^*_i$ for the instances with 
 the full realizations in $\Psi_{\phi_i}$.
 Thus, the second term in the right-hand side of \eqref{eq.induction}
is bounded by
\begin{align*}
 &2\alpha (k-1) \sum_{i=1}^l p(\phi_i) \favg(\pi^*(U_i,\phi_i))\\
 & \leq 
 2\alpha(k-1) \sum_{i=1}^l p(\phi_i) \favg(\pi^*_i)
 \leq 2\alpha (k-1)\favg(\pi^*),
\end{align*}
where the second inequality follows from Lemma~\ref{lem.opt_part}.
Therefore, we have 
$\favg(\pi) \leq 2\alpha k \favg(\pi^*)$.
\end{proof}

\subsection{Algorithm and analysis for the local feedback model}
\label{sec.localfeedback}

Our algorithm and its analysis for the local feedback model are almost
the same as those for the full feedback model. In this subsection, we only
highlight the differences between the algorithms.

As in Section~\ref{subsec.algorithm-full}, let $U$ and $\phi$ be the
set of chosen nodes and the realization kept at the beginning of some round.
In the local feedback model, choosing a node $v$ reveals only
the state of $v$.
Because of this, we possibly choose inactive nodes in this
model.
Due to this difference, we define the most-likely observation vector
$\xi_v$ for all nodes $v \in N[R]\setminus U$ by
\[
 \xi_v(u) =
 \begin{cases}
  1 & u= v \in H,\\
  0 & u= v \not\in H, \\
  * & u \neq v.
 \end{cases}
\]

In each round, we solve two instances of the polymatroid Steiner tree
problem to compute $\Hcds$ and $\Hsubmod$.
The graphs in these instances are $G[N[R]]$.
The ranges of the functions $\fcds$ and $\fsubmod$
are also changed to $2^{N[R]}$.
The reason for these changes
is that an optimal policy
may choose inactive nodes even
in the case that the set of active nodes is exactly $H$.
The definition of $\fcds$ is changed to
\begin{align*}
 \fcds(X)=&|(V\setminus \supp(\phi))\cap N[X \cap R]\}|\\
 & \hspace*{1em} +
\sum_{v \in V\setminus N[R]}
\left(1- \sum_{\psi \in \Psi_{X\setminus U}: v \sim_{\psi} r}p(\psi)\right)
\end{align*}
for each $X \subseteq N[R]$.
Notice that the first term in the right-hand side
depends only on $X \cap R$,
but the second term depends on the whole  of $X$.
This is because inactive nodes are not used to dominate other nodes,
but the information on the states of those nodes indicates that 
the full realizations inconsistent with that information
are not realized.
The definition of $\fsubmod$ is changed to
\[
\fsubmod(X)
=\min\left\{\frac{1}{2},1-\sum_{\psi \in \Psi_{X\setminus U}}p(\psi\mid \phi)\right\}
\]
for each $X \subseteq N[R]$.

The other inputs in
these instances are the same as those in the full feedback model.
$\Hcds$ and $\Hsubmod$ are defined as the sets of nodes in
$N[R]\setminus U$ that are spanned by the computed trees.
Then, the policy repeats nodes in $H^*=\argmin_{H' \in
\{\Hcds,\Hsubmod\}}w(H')$.
If the revealed state of a chosen node $v$ is different from the
expectation given by $H$ (i.e., $v$ is active and $v \not\in H$, or $v$
is inactive and $v \in H$), then the policy stops choosing the nodes in this
round and proceeds to the next round.
The full details of this algorithm are shown in
Algorithm~\ref{alg.local}.

\begin{algorithm}[t]
  \caption{$O(\alpha \log 1/\delta)$-approximation algorithm for the local feedback model}
 \label{alg.local}
 \SetKwInOut{Input}{Input}\SetKwInOut{Output}{Output}
 \Input{undirected graph $G=(V,E)$, node weights $w\colon V\rightarrow
 \Rset_+$, full realizations $\Psi$, probability distribution $p\colon
 \Psi\rightarrow [0,1]$, and root node $r \in V$ ($\psi(r)=1$ for all
 $\psi\in \Psi$)}
 \Output{$U \subseteq V$}
 \BlankLine
 $U \longleftarrow \{r\}$, $\phi(v) \longleftarrow 1$\;
 \lForEach{$v \in V\setminus \{r\}$}{$\phi(v) \longleftarrow *$}
 \While{$U$ does not dominate $V$}{
 $H \longleftarrow \{v \in V\setminus U\colon \sum_{\psi\in \Psi_{\phi}:v\in
 A_{\psi}} p(\psi\mid \phi) > 1/2\}$\;
 $R \longleftarrow$ connected component including $r$ in $G[H
 \cup U]$\;
 compute $\Hcds, \Hsubmod \subseteq N[R]\setminus U$,
 and set $H^* \longleftarrow \argmin_{H' \in \{\Hcds,\Hsubmod\}} w(H')$\;
 \While{$H^* \not\subseteq U$}{
 choose a node $v \in H^* \setminus U$ dominated by $U$\;
 $U\longleftarrow U \cup \{v\}$, $\phi(v)\longleftarrow $ the state of $v$\;
 \lIf{$\phi(v)=1$ and $v \in H$, or $\phi(v)=0$ and $v \not\in H$}
 {break}
 }
 remove all nodes $u$ with $\psi(u)=0$ for all $\psi \in \Psi_{\phi}$
 from $G$\;
 remove the connected components not including $r$ from $G$
 }
 output $U$
\end{algorithm}

The performance guarantee on Algorithm~\ref{alg.local}
is given in the following theorem.
Its proof
is the same as that for Theorem~\ref{thm:approx_ratio}
except several minor
details, and hence we omit it from this paper. 

 \begin{theorem}
  \label{thm:approx_ratio_local}
  Let $\pi$ denote the policy given by Algorithm~\ref{alg.local},
  and let $\alpha$ be the approximation factor for the node-weighted polymatroid Steiner
  tree problem.
  Then, $\pi$ achieves the approximation factor $2\alpha  (1+\lg (1/\delta))$.
 \end{theorem}

 \section{Heuristic algorithms}
 \label{sec.heuristic}

 In this section, we introduce two heuristic algorithms,
 one of which is a greedy algorithm and the other of which is based on
  an algorithm for the CDS problem.
  These algorithms are regarded as baselines in the simulation reported in
  the subsequent section.

 \subsection{Greedy algorithms}
 In the greedy algorithms, we repeat choosing a node for which the
 expected number  of nodes dominated by it is largest.
 We first explain the algorithm for the full feedback model.

 Let $U$ be the set of chosen nodes and $\phi$ be the realization that
 represents  the observed states of nodes up to the beginning of an
 iteration.
 For each $v \in \supp(\phi)\setminus U$ such that $\phi(v)=1$,
 we define its score $\chi(v)$ as 
 $\sum_{u \in N[v] \setminus \supp(\phi)} \Pr[\psi(u)=1
  \mid \psi  \succeq \phi]$.
  Hence, $\chi(v)$ represents the expected number of active nodes that
  are not dominated by $U$ but by $v$.
  In this iteration, the policy chooses the node $v$ that maximizes
  $\chi(v)/w(v)$ among all nodes in $\supp(\phi)\setminus U$.

  Recall that $U= \supp(\phi)$ in the local feedback model.
  For the local feedback model,
  we choose a node that is dominated by an active node in $U$.
  We define the score of such a node $v$
  as 
 $\Pr[\psi(v)=1 \mid \psi \succeq \phi] \times \sum_{u \in N[v] \setminus N[U]} \Pr[\psi(u)=1
  \mid \psi  \succeq \phi, \psi(v)=1]$. Namely, it is
  the expected number of active
  nodes that are not dominated by $U$ but by $v$,
  multiplied by the probability that $v$ is active.

  The full algorithm is shown in Algorithm~\ref{alg.greedy}.
  
\begin{algorithm}[t]
  \caption{Greedy algorithm}
 \label{alg.greedy}
 \SetKwInOut{Input}{Input}\SetKwInOut{Output}{Output}
 \Input{undirected graph $G=(V,E)$, node weights $w\colon V\rightarrow
 \Rset_+$, full realizations $\Psi$, probability distribution $p\colon
 \Psi\rightarrow [0,1]$, and root node $r \in V$ ($\psi(r)=1$ for all
 $\psi\in \Psi$)}
 \Output{$U \subseteq V$}
 \BlankLine
 $U \longleftarrow \{r\}$ and $\phi(r) \longleftarrow 1$\;
 \ForEach{$v \in V$}{
 \lIf{the model is the full feedback model, $v \in N[r]$, and $v$ is active}{$\phi(v)\longleftarrow 1$}
 \lElseIf{the model is the full feedback model, $v \in N[r]$ ($v$ is
 inactive)}{$\phi(v)\longleftarrow 0$}
  \lElse{$\phi(v)\longleftarrow *$}
 }
 \While{$U$ does not dominate $V$}{
 $D\longleftarrow \{v \in \supp(\phi)\setminus U \colon \phi(v)=1\}$
 in the full feedback model\;
 $D \longleftarrow \{v \in V\setminus
 U\colon \exists u \in U, \phi(u)=1, uv \in E\}$
 in the local feedback model\;
 \lForEach{$v \in D$}{compute $\chi(v)$}
 $v \longleftarrow \argmax_{v \in D} \chi(v)/w(v)$ and $U \longleftarrow
 U \cup \{v\}$\;
 update $\phi(v)$ in the full feedback model, and $\phi(u)$ for
 $\forall u \in N[v]\setminus U$ in the local feedback model\;
 remove all nodes $u$ with $\psi(u)=0$ for all $\psi \in \Psi_{\phi}$
 from $G$\;
 remove the connected components not including $r$ from $G$
 }
 output $U$
\end{algorithm}

\paragraph*{Bad instance}
The performance guarantee of Algorithm~\ref{alg.greedy}
can be arbitrarily bad even for the (deterministic) CDS
problem.
Let $G=(V,E)$ be the graph defined as follows (see also Figure~\ref{fig.badex}(a)).
The node set $V$ consists of nodes 
$r,v, u, u'$, and the other $n'$ nodes.
We let $U = V\setminus \{r,v,u,u'\}$.
The edge set $E$ contains 
edges $rv, ru, uu'$.
In addition, each of $v$ and $u'$ is joined with all nodes 
in $U$.
We can see that $G$ is a unit disk graph by defining the positions of
the nodes appropriately.
The node weights $w$ are defined by $w(r)=0$, $w(v)=n'$,
$w(u)=1+\epsilon$ for a small $\epsilon > 0$, $w(u')=0$,
and $w(x)=0$ for each node $x \in U$.
The minimum-weight CDS for this instance is
$\{r,u,u'\}$, and its weight is $1+\epsilon$.
However,
Algorithm~\ref{alg.greedy} outputs another CDS that consists of $r,v$, and a
node of $U$.
The weight of this CDS is $n'$.
The factor of the CDS output by 
Algorithm~\ref{alg.greedy} 
to the optimal one
is $n'/(1+\epsilon)$, which can be arbitrarily large by setting  $\epsilon$ 
to a value near to $0$ and by increasing $n'$.

 \begin{figure}
  \centering
 \subfigure[greedy algorithm]{\includegraphics[]{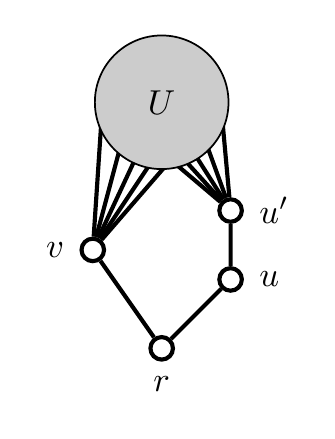}}
 \subfigure[CDS-based algorithm]{\includegraphics[]{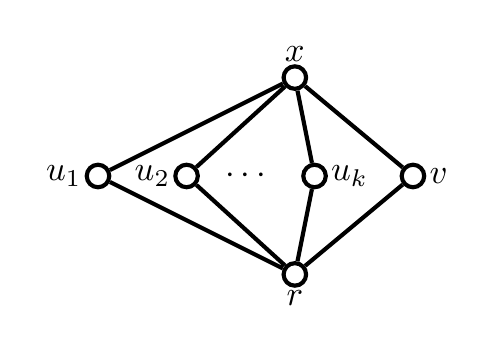}}
\caption{Bad instances for the greedy and the CDS-based algorithms}
\label{fig.badex}
 \end{figure}

  \subsection{CDS-based algorithms}
  In the other algorithm, we use a $\beta$-approximation algorithm
  for the CDS problem.
  In this algorithm, we first compute a CDS $D$ for the graph $G$
  and the node weights $w$
  by the $\beta$-approximation algorithm.
  Then, we repeat choosing nodes in $D$ such that the chosen nodes
  always induce connected subgraphs.
  During this process, if a node in $D$ turns out to be inactive,
  then we stop choosing nodes in $D$. We remove the nodes that are revealed
  to be inactive from $G$ and recompute a CDS for the updated graph,
  where the weights of the nodes that have been chosen up to that moment
  are defined to be $0$.
  These operations are repeated until the chosen nodes form a solution.
  The pseudo-code of this algorithm is given in 
  Algorithm~\ref{alg.cds-based}.
  
 \begin{algorithm}[t]
  \caption{CDS-based algorithm}
 \label{alg.cds-based}
 \SetKwInOut{Input}{Input}\SetKwInOut{Output}{Output}
 \Input{undirected graph $G=(V,E)$, node weights $w\colon V\rightarrow
 \Rset_+$, full realizations $\Psi$, probability distribution $p\colon
 \Psi\rightarrow [0,1]$, and root node $r \in V$ ($\psi(r)=1$ for all
 $\psi\in \Psi$)}
 \Output{$U \subseteq V$}
 \BlankLine
 $U \longleftarrow \{r\}$ and $\phi(r) \longleftarrow 1$\;
 \ForEach{$v \in V$}{
 \lIf{the model is the full feedback model, $v \in N[r]$, and $v$ is active}{$\phi(v)\longleftarrow 1$}
 \lElseIf{the model is the full feedback model, $v \in N[r]$ ($v$ is
 inactive)}{$\phi(v)\longleftarrow 0$}
  \lElse{$\phi(v)\longleftarrow *$}
 }
  \While{$U$ does not dominate $V$}{
  update $w(u) \longleftarrow 0$ for $\forall u \in U$\;
  compute a minimum-weight CDS $D$ of $(G,w)$\;
  \While{$(D \cap N[U])\setminus U \neq \emptyset$}{
  choose $v \in (D \cap N[U])\setminus U$ and $U \longleftarrow U\cup
  \{v\}$\;
  update $\phi(v)$ in the full feedback model, and $\phi(u)$ for
  $\forall u \in N[v]\setminus U$ in the local feedback model\;
  \lIf{$\phi(u)=0$ for some node $u$}{break}
  }
  remove all nodes $u$ with $\psi(u)=0$ for all $\psi \in \Psi_{\phi}$
  from $G$\;
  remove the connected components not including $r$ from $G$
 }
 output $U$
\end{algorithm}

\paragraph*{Bad instance}
Algorithm~\ref{alg.cds-based}
does not perform well for some instances
because its behavior does not 
depend on the probability distribution on the node states.
We see this fact by presenting a bad
instance
for Algorithm~\ref{alg.cds-based}.
For ease of discussion, we suppose that a minimum-weight CDS
can be computed in Step~8 of Algorithm~\ref{alg.cds-based},
and consider the local feedback model.

Let $\epsilon$ and $\delta$ be small positive constants.
Let $G=(V,E)$ be the graph illustrated in Figure~\ref{fig.badex}(b).
Namely, the node set $V$ of a graph $G$
comprises nodes $r,v,u_1,\ldots,u_k,x$,
and the edge set $E$ consists of 
edges $rv, ru_1,\ldots,ru_k, xv, xu_1,\ldots,xu_k$.
The node weights of $u_1,\ldots,u_k$ are $1$, and 
that of $v$ is $1+\epsilon$,
and those of $r$ and $x$ are 0.
We assume that nodes $r$, $x$, and $v$ are always active.
For a full realization $\psi \in \Psi$,
$p(\psi)=\delta$ if 
exactly one of $u_1,\ldots,u_k$ 
is active,
and 
$p(\psi)=1-k \delta$ if 
all of $u_1,\ldots,u_k$ 
are inactive.

If $\delta$ is small enough, the average objective value achieved 
by a policy is almost equal to the one for the case where all of $u_1,\ldots,u_k$ are inactive.
In this case, 
the best policy chooses only $v$, which results in the objective value $1+\epsilon$.
However, Algorithm~\ref{alg.cds-based} chooses all of $u_1,\ldots,u_k$ and then chooses $v$.
Thus the approximation factor of Algorithm~\ref{alg.cds-based} is nearly $k/(1+\epsilon)$.
This factor can be arbitrarily large by increasing $k$
and setting $\epsilon$ to a constant.

Notice that the above discussion holds even when there are edges induced by $\{v,u_1,\ldots,u_k\}$. Thus we can find such a bad instance even if the graph is restricted to unit disk graphs. Moreover, we obtain a bad example in the full feedback model by modifying the above instance slightly as follows. 
For each $i\in [k]$,
subdivide the edge $ru_i$ by inserting a new node $u'_i$ 
(i.e., edge $ru_i$ is replaced by $ru'_i$ and $u'_i u_i$),
where we let $u'_i$ be always active. 
The weights of $u'_1,\ldots,u'_k$ are defined as $1$, and those
of $u_1,\ldots,u_k$ are changed to 0.
Then, similar to the above instance for the local feedback model,
 Algorithm~\ref{alg.cds-based} chooses all of $u'_1,\ldots,u'_k$ and $v$,
achieving the objective $k+1+\epsilon$
while the optimal policy chooses $v$ and $x$, achieving the objective $1+\epsilon$.

 \section{Simulation results}
\label{sec.simulation}

In this section, we report our simulation results to evaluate the empirical
performance of the proposed algorithm.

\subsection{Setting}
Throughout this section, we call the algorithm given in Section~\ref{sec.fullfeedback}
the polymatroid-based algorithm,
and the heuristic algorithms given in Section~\ref{sec.heuristic}
the greedy and the CDS-based algorithms.
In the implementations of the polymatroid-based and the CDS-based algorithms, we use the CBC integer programming solver (\url{https://projects.coin-or.org/Cbc})
to solve the polymatroid Steiner tree and the CDS problems. The solver
solves these problems to optimality in all of the simulations that we
report here. Thus, the polymatroid-based algorithm achieves an approximation factor within $1+\lg(1/\delta)$ here.

In addition to these algorithms, we implemented 
the local algorithm given by Borgs\ et\ al.~\cite{BorgsBCKL12}.
The reason to choose this algorithm is that this algorithm requires the most limited information among the known local algorithms;
The other local algorithms \cite{GuhaK98,KhullerY16} require the information on the 3-hop neighbors of the previously chosen 
nodes while the algorithm of Borgs\ et\ al.\ requires only the information on the 2-hop neighbors.
Note that even in the full feedback models, our algorithms cannot observe the 2-hop neighbors, and thus 
the algorithm of Borgs\ et\ al.\ uses more information on neighbors than our algorithms.
To be fair, our algorithms use the probability distributions on realizations instead.

The algorithm of Borgs\ et\ al.\ is defined only for the cardinality minimization. To apply it to the weight minimization,
we slightly modified the algorithm. 
Let us explain this modification briefly.
The algorithm of Borgs\ et\ al.\ maintains a solution set $U$ in the course of its execution,
and proceeds in iterations.
In each iteration, it first chooses a node $v$ from $N[U]$ that maximizes $|N[v]\setminus N[U]|$, and add $v$ to $U$.
Then, it also chooses a node from $N[v]\setminus N[U]$ uniform randomly, and adds it to $U$ if it dominates a node not dominated by $U$.
In our implementation, 
we define $v$ as a node in $N[U]$ that maximizes $|N[v]\setminus N[U]|/w(v)$ instead of $|N[v]\setminus N[U]|$.
The other part of the algorithm is the same as the original.

For the simulations, we prepared the following three types of instances:
\paragraph*{Unit disk graph}
 In this type of instances,
 the graph is a unit
 disk graph,
 where each node has a random position
 distributed uniformly on a 2-dimensional square region of a unit side length,
 and two nodes are joined by an edge if their distance is at most
 $1/\sqrt{n}$. 
 We intend to make the density of the graph uniform when $n$ changes
 by this definition of edges.
 The weight of a node is sampled uniformly
 at random from $[0,1]$. 
 As for the probability distribution on the full realizations,
 we construct $M$ full realizations as follows:
 for each full realization,
  we randomly pick seven points on the square region
  with thresholds sampled from $[0,1/3]$,
  and a node except the root and its neighbors
  is defined to be inactive if and only if
  its distance from one of the seven points
  is at most its threshold. The root and its neighbors are always active.
  The distribution on the full realizations is defined
  so that one of these $M$ full realizations
  appears uniformly at random.
  Therefore, $\delta = 1/M$ in these instances.

 \paragraph*{Bidirectional disk graph}
 In graphs of these instances,  each node has a random position on the
 a 2-dimensional square region of a unit side length,
 and has a radius sampled uniform randomly from $[0,1/3]$.
 We consider these graphs as a model of wireless networks;
 the radius of a node represents a coverage length of a wireless device,
 and an edge indicates that the devices corresponding to its end nodes
 can send signals each other.
 Two nodes are joined by an edge whenever their distance is at
 most the minimum of their radiuses.
 The other settings are defined in the same way as the unit disk graph instances.

 \paragraph*{Erd\H{o}s-R\'enyi random graph}
 In this type of instances, the graph is an Erd\H{o}s-R\'enyi random graph, which includes an edge joining each pair of nodes with probability $0.1$ independently.
 The node weights are defined in the same way as the unit disk graph instances.
 To define the full realizations, we define the location of each node on a 2-dimensional square region of a unit side length randomly, 
and construct $M$ full realizations and a distribution on them 
 as in the unit disk graph instances. 
 The thresholds associated with the seven points picked for defining the full realizations are sampled from $[0,0.25]$ since the the sampling from $[0,1/3]$ makes the graphs on the active nodes too small.

\subsection{Results}	      

\paragraph{Comparison of the adaptive algorithms}
As we noted in Section~\ref{sec.heuristic},
the approximation factors of the heuristic algorithms 
are unbounded for the carefully constructed instances.
However, these algorithms possibly perform well for other instances. 
A purpose of this simulation is to 
compare their performances and that of the polymatroid-based algorithm.

In Figure~\ref{fig.simulation}, we compare the average objective values
achieved by the algorithms. 
Each point in the figure corresponds to an
instance of the problem, and its $x$-coordinate (resp., $y$-coordinate)
is the ratio of the average objective value achieved by the
polymatroid-based algorithm to that achieved by the greedy algorithm
(resp., the CDS-based algorithm).
Thus, if the $x$-coordinate (resp., $y$-coordinate) of the point
is smaller than 1, then the polymatroid-based algorithm performs better than
the CDS-based algorithm (resp., the greedy algorithm) in the
corresponding instance.
If the point is below the line representing $x=y$,
then this indicates that the CDS-based algorithm is better than
the greedy algorithm.
In the construction of the instances, 
parameters are set to $M=30$ and $n=40,60,80,100$.
For each instance type and each parameter setting,
we constructed 5 instances.

We can observe that all of the three algorithms achieve the best
in some instances, and thus we can find no algorithm that is superior or is inferior to the other algorithms completely.
In the full feedback model, the polymatroid-based and the CDS-based  algorithms are superior to the greedy algorithm in many of the instances.
Comparing the polymatroid-based and the CDS-based  algorithms,
the polymatroid-based algorithm is slightly better than the CDS-based algorithm, but their difference is small.
In the local feedback model, the difference is clear; the polymatroid-based algorithm
performs better than the other two algorithms in unit disk and the Erdos-Renyi graphs. 
For bidirected disk graphs with the local feedback model,
the CDS-based algorithm is superior to the others.

Although the merit of the greedy algorithm may not be clear when the objective values are compared, it is superior to the other algorithms in running time.
For example, the greedy algorithm
is 10000 times faster than the other two algorithms for the unit disk graph instances with $n=100$ in our simulations.
Our implementations of the polymatroid-based and the CDS-based algorithms have to solve the integer programs, and the time for this part is dominant in their computations.
The running times of these
algorithms can be improved by more sophisticated integer programming solvers
or by implementing other algorithms for the polymatroid Steiner tree and the CDS problems
at the sacrifices of their optimization performances.
However, 
it is unlikely that they are superior to
the greedy algorithm because of their definitions.
Thus the greedy algorithm is a good option when 
a fast algorithm is required (e.g., when the network is frequently updated).
On the other hand, it is often the case that receiving feedback
from the environment is time-consuming. In such a situation,
the running time of an algorithm is not an important factor, and the
polymatroid-based and the CDS-based algorithms are also useful.

\paragraph{Comparison of the adaptive algorithms and the local algorithm}

We also apply the implementation of the local algorithm of 
Borgs\ et\ al.~\cite{BorgsBCKL12} to the instances used in the simulation reported above.
In all of the instances, the local algorithm was clearly inferior to our three algorithms.
In Table~\ref{table.local}, we report 
the maximum and minimum ratios of the objective values achieved by the local algorithm
to that achieved by the polymatroid-based algorithm for each graph class.

\paragraph{Change of performances for varying $M$}
The approximation factor of the polymatroid-based algorithm
depends on $1/\delta$ ($= M$) in our theoretical guarantee;
as $1/\delta$ increases, the approximation factor gets worse.
To evaluate the influence of $1/\delta$ on the performance of the algorithms,
we did the simulations with varying $M$.

Figure~\ref{fig.simulation2} indicates how performances 
of the algorithms change when $M$ changes.
We used the unit disk graph instances with $n=70$.
For each $M$, we prepared 5 instances,
and report the averages of the objective values 
achieved by the algorithms for those 5 instances.
The right figure in Figure~\ref{fig.simulation2} 
also gives the average weights of active nodes
in the solutions output by the algorithms.
Recall that a solution includes both active and inactive nodes in the local feedback model
while a solution consists of only active nodes in the full feedback model.
Although the objective in the local feedback models is not to minimize the weights of the active nodes in a solution, it is meaningful to see how large they are since the active nodes in the solution forms a CDS.

The figure shows that the polymatroid-based algorithm is superior to the other algorithms
for all values of $M$.
We cannot observe any influence from increasing $M$ on the
performance in the full feedback model.
In the local feedback model, the average weights are increased 
as $M$ increases, but this is a common phenomenon in all algorithms.
Hence, 
at least in this setting,
increasing $M$ does not give a large influence on the
performance of the polymatroid-based algorithm.

From the right figure of Figure~\ref{fig.simulation2},
we can see that the 
weights of the solutions output by
the CDS-based algorithm is larger than those of the greedy algorithm
in the local feedback model.
 However, the weights of the active nodes in the solutions
output by the CDS-based algorithm is smaller than those of the greedy algorithm.
This is an interesting feature of the CDS-based algorithm.

 \begin{figure}[t]
  \centering
\includegraphics[scale=.5]{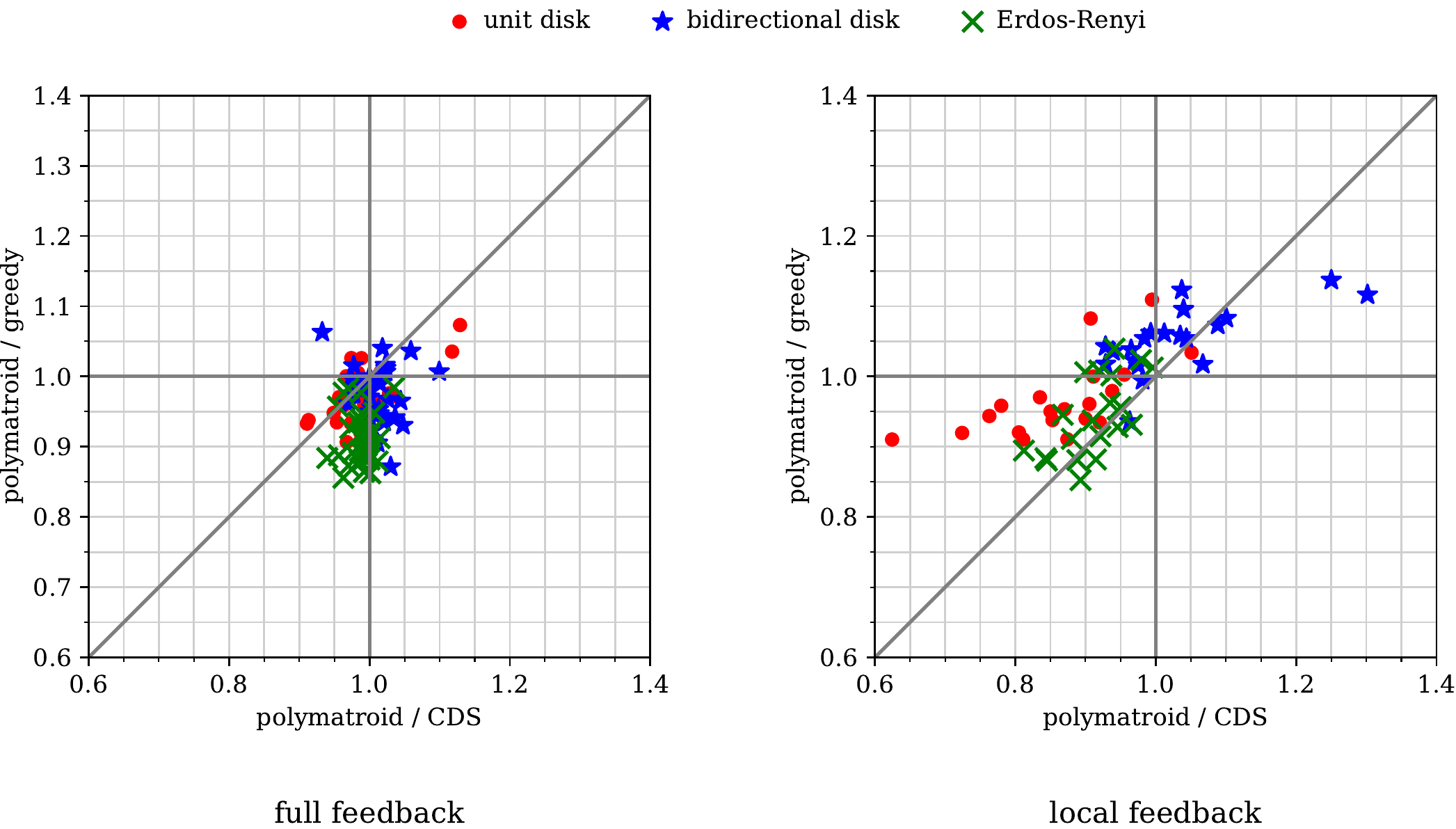}
 \caption{Simulation results to compare the three proposed algorithms}
  \label{fig.simulation}
 \end{figure}
\begin{table}[t]
\caption{Ratios of objective values achieved by Borgs\ et\ al.~\cite{BorgsBCKL12} to the polymatroid-based algorithm}
\label{table.local}
\centering
 \begin{tabular}{|c|ccc|}
  \hline
  & unit disk & bidirectional disk & Erdos-Renyi\\
  \hline 
  maximum & 1.588 & 1.494 & 2.357\\
  minimum & 1.359 & 1.300 & 1.483\\  \hline
 \end{tabular}
\end{table}

 \begin{figure}[t]
  \centering
  \includegraphics[scale=.4]{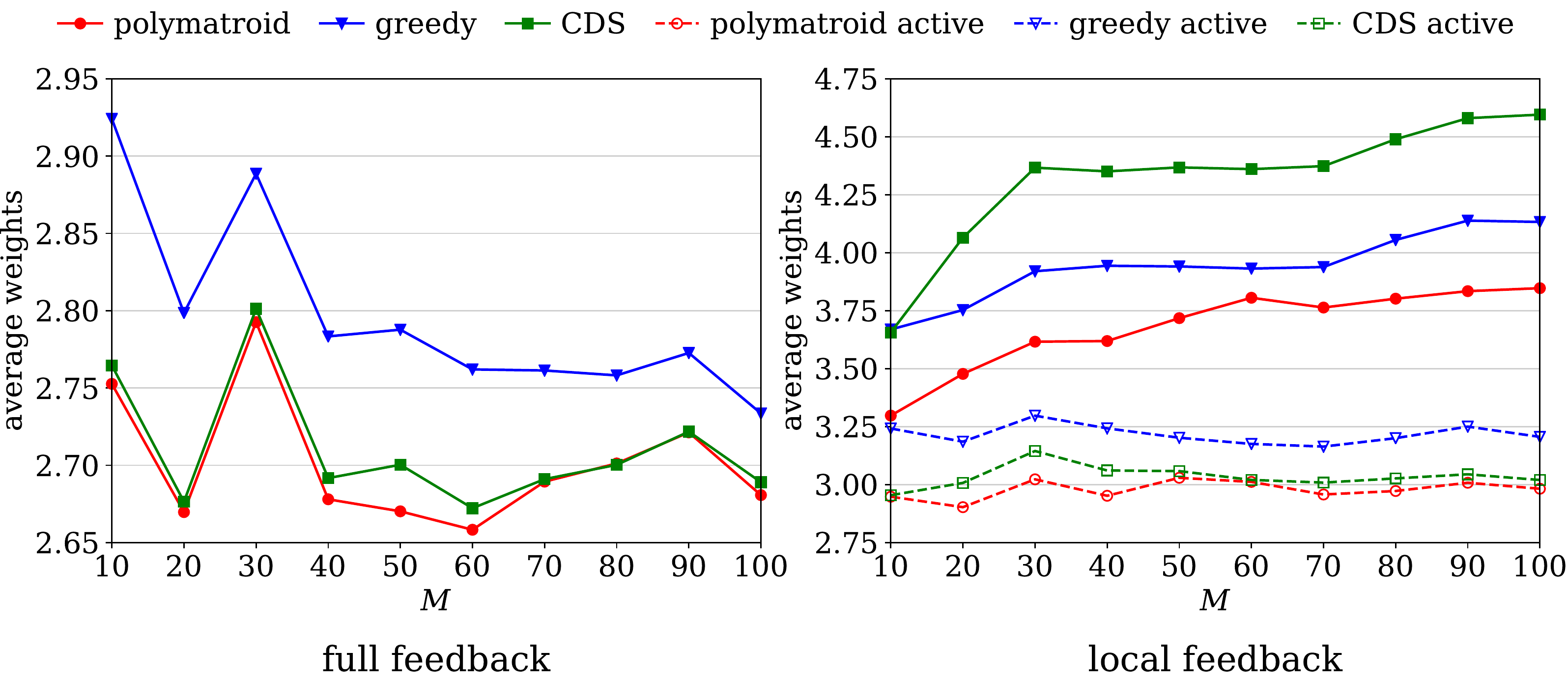}
  \caption{Average objective value as a function of $M$ in the
  unit disk graph instances}
  \label{fig.simulation2}
 \end{figure}

\section{Conclusion}
\label{sec.conclusion}

In this paper, we formulate the robust CDS problem, which is a new
stochastic optimization variant of the CDS problem.
We also present three algorithms. One of them has a theoretical
performance guarantee of factor $O(\alpha \log 1/\delta)$.
Through computational simulation, we compare their empirical
performances.
Considering
the instability of wireless ad hoc networks, 
we believe that these algorithms are useful for the construction of the virtual
backbone of a network.


\begin{thebibliography}{10}
\providecommand{\url}[1]{#1}
\csname url@samestyle\endcsname
\providecommand{\newblock}{\relax}
\providecommand{\bibinfo}[2]{#2}
\providecommand{\BIBentrySTDinterwordspacing}{\spaceskip=0pt\relax}
\providecommand{\BIBentryALTinterwordstretchfactor}{4}
\providecommand{\BIBentryALTinterwordspacing}{\spaceskip=\fontdimen2\font plus
\BIBentryALTinterwordstretchfactor\fontdimen3\font minus
  \fontdimen4\font\relax}
\providecommand{\BIBforeignlanguage}[2]{{%
\expandafter\ifx\csname l@#1\endcsname\relax
\typeout{** WARNING: IEEEtran.bst: No hyphenation pattern has been}%
\typeout{** loaded for the language `#1'. Using the pattern for}%
\typeout{** the default language instead.}%
\else
\language=\csname l@#1\endcsname
\fi
#2}}
\providecommand{\BIBdecl}{\relax}
\BIBdecl

\bibitem{Dai:2006}
F.~Dai and J.~Wu, ``On constructing k-connected k-dominating set in wireless ad
  hoc and sensor networks,'' \emph{J. Parall. Distri. Comp.}, vol.~66, no.~7,
  pp. 947--958, 2006.

\bibitem{Shang:2007jg}
W.~Shang, F.~Yao, P.~Wan, and X.~Hu, ``On minimum m-connected k-dominating set
  problem in unit disc graphs,'' \emph{J. Comb. Optim.}, vol.~16, no.~2, pp.
  99--106, 2008.

\bibitem{ShiZMD17}
\BIBentryALTinterwordspacing
Y.~Shi, Z.~Zhang, Y.~Mo, and D.~Du, ``Approximation algorithm for minimum
  weight fault-tolerant virtual backbone in unit disk graphs,''
  \emph{{IEEE/ACM} Trans. Netw.}, vol.~25, no.~2, pp. 925--933, 2017. 

\bibitem{Wang:2015}
\BIBentryALTinterwordspacing
W.~Wang, B.~Liu, D.~Kim, D.~Li, J.~Wang, and Y.~Jiang, ``A better constant
  approximation for minimum 3-connected m-dominating set problem in unit disk
  graph using {T}utte decomposition,'' in \emph{{IEEE INFOCOM}}, 2015, pp.
  1796--1804. 

\bibitem{Zhang16}
\BIBentryALTinterwordspacing
Z.~Zhang, J.~Zhou, Y.~Mo, and D.~Du, ``Performance-guaranteed approximation
  algorithm for fault-tolerant connected dominating set in wireless networks,''
  in \emph{{IEEE INFOCOM}}, 2016, pp. 1--8. 

\bibitem{GuhaK98}
\BIBentryALTinterwordspacing
S.~Guha and S.~Khuller, ``Approximation algorithms for connected dominating
  sets,'' \emph{Algorithmica}, vol.~20, no.~4, pp. 374--387, 1998.

\bibitem{GuhaK99}
------, ``Improved methods for approximating node weighted {S}teiner trees and
  connected dominating sets,'' \emph{Inf. Comput.}, vol. 150, no.~1, pp.
  57--74, 1999.

\bibitem{ClarkCJ90}
\BIBentryALTinterwordspacing
B.~N. Clark, C.~J. Colbourn, and D.~S. Johnson, ``Unit disk graphs,''
  \emph{Disc. Math.}, vol.~86, no. 1-3, pp. 165--177, 1990. 

\bibitem{ChengHLWD03}
\BIBentryALTinterwordspacing
X.~Cheng, X.~Huang, D.~Li, W.~Wu, and D.~Du, ``A polynomial-time approximation
  scheme for the minimum-connected dominating set in ad hoc wireless
  networks,'' \emph{Networks}, vol.~42, no.~4, pp. 202--208, 2003.

\bibitem{WillsonZWD15}
\BIBentryALTinterwordspacing
J.~Willson, Z.~Zhang, W.~Wu, and D.~Du, ``Fault-tolerant coverage with maximum
  lifetime in wireless sensor networks,'' in \emph{{Proceedings of IEEE International
	Conference on Computer Communication, INFOCOM}}, 2015, pp.
  1364--1372. 

\bibitem{ZouLGW09}
\BIBentryALTinterwordspacing
F.~Zou, X.~Li, S.~Gao, and W.~Wu, ``Node-weighted {S}teiner tree approximation
  in unit disk graphs,'' \emph{J. Comb. Optim.}, vol.~18, no.~4, pp. 342--349,
	2009.
	
\bibitem{ByrkaGRS13}
\BIBentryALTinterwordspacing
J.~Byrka, F.~Grandoni, T.~Rothvo{\ss}, and L.~Sanit{\`{a}}, ``{S}teiner tree
  approximation via iterative randomized rounding,'' \emph{J. {ACM}}, vol.~60,
  no.~1, pp. 6:1--6:33, 2013.

\bibitem{BorgsBCKL12}
C.~Borgs, M.~Brautbar, J.~T. Chayes, S.~Khanna, and B.~Lucier, ``The power of
  local information in social networks,'' in \emph{Proceedings of 8th
	International Workshop on Internet and Network
  Economics, {WINE}}, 2012, pp. 406--419.

\bibitem{KhullerY16}
S.~Khuller and S.~Yang, ``Revisiting connected dominating sets: An optimal
  local algorithm?'' in \emph{Approximation, Randomization, and Combinatorial
  Optimization. Algorithms and Techniques, {APPROX/RANDOM}, {LIPIcs} 60}, 2016,
  pp. 11:1--11:12.

\bibitem{funke2006simple}
S.~Funke, A.~Kesselman, U.~Meyer, and M.~Segal, ``A simple improved distributed
  algorithm for minimum {CDS} in unit disk graphs,'' \emph{{ACM} Transactions
  on Sensor Networks}, vol.~2, no.~3, pp. 444--453, 2006.

\bibitem{ShiCCL16}
T.~Shi, S.~Cheng, Z.~Cai, and J.~Li, ``Adaptive connected dominating set
  discovering algorithm in energy-harvest sensor networks,'' in
	\emph{Proceedings of {IEEE
  International Conference on Computer Communications, INFOCOM}}, 2016, pp. 1--9.

\bibitem{LimHL17}
Z.~W. Lim, D.~Hsu, and W.~S. Lee, ``Shortest path under uncertainty:
  Exploration versus exploitation,'' in \emph{Proceedings of Conference
	on Uncertainty in Artificial Intelligence, {UAI}}, 2017.

\bibitem{LimHL15}
------, ``Adaptive stochastic optimization: From sets to paths,'' in
  \emph{Proceedings of 29th Conference on Neural Information Processing
	Systems, {NIPS}}, 2015, pp. 1585--1593.

\bibitem{GuptaNR17}
A.~Gupta, V.~Nagarajan, and R.~Ravi, ``Approximation algorithms for optimal
  decision trees and adaptive {TSP} problems,'' \emph{Math. Oper. Res.},
  vol.~42, no.~3, pp. 876--896, 2017.

\bibitem{CalinescuZ05}
G.~C{\u{a}}linescu and A.~Zelikovsky, ``The polymatroid {S}teiner problems,''
  \emph{J. Comb. Optim.}, vol.~9, no.~3, pp. 281--294, 2005.

\end{thebibliography}
\end{document}